\documentclass[letterpaper, 10pt, conference]{IEEEtran}
\usepackage[cmex10]{amsmath}
\usepackage{amsthm, amssymb, cite, graphicx, epsfig, url, color, setspace,epstopdf,subfigure}
\usepackage{algorithm}
\usepackage{algpseudocode}
\usepackage{verbatim}

\newcommand{\ignore}[1]{}

\newtheorem{theorem}{Theorem}[section]

\newcommand\mygeqa{\mathrel{\overset{\makebox[0pt]{\mbox{\normalfont\tiny\sffamily (a)}}}{\geq}}}
\newcommand\mygeqb{\mathrel{\overset{\makebox[0pt]{\mbox{\normalfont\tiny\sffamily (b)}}}{\geq}}}
\newcommand\myeqa{\mathrel{\overset{\makebox[0pt]{\mbox{\normalfont\tiny\sffamily (a)}}}{=}}}

\newcommand\myleqa{\mathrel{\overset{\makebox[0pt]{\mbox{\normalfont\tiny\sffamily (a)}}}{\leq}}}
\newcommand\myleqb{\mathrel{\overset{\makebox[0pt]{\mbox{\normalfont\tiny\sffamily (b)}}}{\leq}}}

\newcommand\myrightarrowa{\mathrel{\overset{\makebox[0pt]{\mbox{\normalfont\tiny\sffamily (a)}}}{\Rightarrow}}}

\IEEEoverridecommandlockouts

\begin{document}

\title{An Online Algorithm for Power-proportional Data Centers with Switching Cost}

\author{\IEEEauthorblockN{Ming Zhang}
\IEEEauthorblockA{Department of CSE\\
The Ohio State University\\}
\and
\IEEEauthorblockN{Zizhan Zheng}
\IEEEauthorblockA{Department of Computer Science\\
Tulane University\\}
\and
\IEEEauthorblockN{Ness B. Shroff}
\IEEEauthorblockA{Department of ECE and CSE\\
The Ohio State University}
}

\maketitle

\begin{abstract}
Recent studies have shown that power-proportional data centers can save energy cost by dynamically ``right-sizing'' the data centers based on real-time workload. More servers are activated when the workload increases while some servers can be put into the sleep mode during periods of low load. In this paper, we revisit the dynamic right-sizing problem for heterogeneous data centers with various operational cost and switching cost. We propose a new online algorithm based on a regularization technique, which achieves a better competitive ratio compared to the state-of-the-art greedy algorithm in \cite{lin2012online}. We further introduce a switching cost offset into the model and extend our algorithm to this new setting. Simulations based on real workload and renewable energy traces show that our algorithms outperform the greedy algorithm in both settings.
\end{abstract}

\setlength{\textfloatsep}{20pt}
\section{Introduction}\label{sec:intro}
Internet-scale services like web-mail, live streaming, online gaming and social networks usually have millions or even billions of active users everyday. Providers like Google, Amazon and Facebook, in order to maintain the reliability, accessibility and guaranteed performance of their systems, have deployed numerous large data centers including massive number of servers, causing a huge amount of electricity and cooling cost. Based on \cite{koomey2011growth}, the electricity consumption of large data centers has accounted for $1.3\%$ of all the electricity use of the world and almost $2\%$ of the United States in 2010.

Recent research \cite{qureshi2009cutting,pakbaznia2009minimizing,rao2010minimizing,stanojevic2010distributed} shows that the energy cost can be significantly reduced by dynamically distributing the workload to various data centers based on the idea of ``Geographical Load Balancing'' (GLB) and ``Right-sizing'' to make the data center more power-proportional \cite{chase2001managing,lin2013dynamic,abts2010energy}. Specifically, the central load balancer dynamically dispatches the workload requests to geographically located data centers that consist of thousands of servers. Each data center dynamically adjusts the number of active servers to serve the requests so that during low-load period, servers that do not have jobs transfer to the power-saving mode or are shut down completely after data and operation states are reserved.

In \cite{lin2012online,lin2013dynamic}, Lin et al., propose a cost minimizing model for the ``right-sizing'' of data centers incurring both operational cost and switching cost. Their model is a general convex optimization problem where the objective function consists of two parts representing the operational cost and switching cost, respectively. Examples are given to show how concrete energy and delay costs of data centers can fit into their model. The operational cost is modeled as a time-dependent convex function and a linear function is used to represent the switching cost of switching certain number of servers from power-saving mode to active mode to serve the increasing workload \cite{barroso2007case}. The switching cost is incurred only when the number of active servers increase. Such switching cost not only includes the total energy cost, but also delay in data migration, increased wear-and-tear on servers and the risk involved in server mode toggling. A 3-competitive online algorithm \cite{lin2013dynamic} is proposed for the case of a single server. In \cite{lin2012online}, Lin et al. consider a more general setting with multiple heterogeneous servers and look-ahead information, and propose the AFHC online algorithm that is $(1+\frac{\beta}{e_0})$-competitive where $\beta$ is the maximum value of unit switching cost and $e_0$ refers to the minimum unit operational cost of all data centers. The algorithm reduces to the simple greedy algorithm when the look-ahead window size is zero. The two online algorithms are the first attempts to deal with the online convex right-sizing problem with switching cost and provide performance guarantees. However, the 3-competitive algorithm only works for the single-server setting, while the greedy algorithm can have a large competitive ratio when the minimum value of the unit operating cost $e_0$ is very small compared to the switching cost.

In this paper, we revisit the right-sizing problem studied in~\cite{lin2012online} and propose a better algorithm. We consider a system with multiple data centers located in different places. The operational cost and switching cost of each data center vary based on the local energy prices, the availability of renewable energy, and other factors such as energy storage and servers' wear-and-tear cost. The information of the workload and cost functions of each data center are both revealed only at the beginning of each time slot. There is no look-ahead window, meaning that information for future time slots is not available at the central load balancer. Real-time workload demand is dispatched to different data centers at each time slot by the central load balancer which tries to minimize the total cost for all time slots.

We develop a new online algorithm based on the regularization technique proposed in \cite{buchbinder2014competitive} for the right-sizing problem. We show that our online regularization algorithm achieves a better competitive ratio compared to the greedy algorithm in~\cite{lin2012online}. We further extend our system model by introducing a new time-dependent parameter, called switching cost offset, which enables a data center to serve the increase in workload demand without incurring any switching cost when the increase is less than the offset parameter. This new parameter is meaningful since each data center may have access to some local renewable energy source or energy storage, which allows it to activate some number of sleeping servers by paying a negligible cost. In addition, the switching cost offset also includes the delay tolerance during data migration when servers are activated or turned down and the cost compensation by certain protection mechanism to reduce the server state toggling cost. To the best of our knowledge, this is the first work that considers such a switching cost offset. We propose another online regularization algorithm with guaranteed performance. 

Our main contribution can be summarized as follows:
\begin{enumerate}
\item We propose an online regularization algorithm for the right-sizing problem of multiple heterogenous data centers with various operational cost and switching cost. We prove a competitive ratio of our algorithm in terms of the switching cost and the operational cost functions, which is always better than that of the greedy algorithm in~\cite{lin2012online}.
\item We consider an extension of the right-sizing problem by introducing a switching cost offset into the model, and propose an online algorithm with guaranteed performance. Our algorithm is the first attempt to deal with this important extension.
\end{enumerate}

The rest of the paper is organized as follow: Section~\ref{sec:related} discusses the related work on energy cost minimization of data centers and the regularization method. Section~\ref{sec:model} introduces our general system model and the online algorithm is proposed in Section~\ref{sec:onlineconvex}. We then discuss the model with switching cost offset and the corresponding algorithm in Section~\ref{sec:reg}. Numerical results are given in Section~\ref{sec:num} and we conclude the paper in Section~\ref{sec:conclusion}.

\section{Related Work}\label{sec:related}
There are multiple recent works \cite{chen2005managing,gandhi2009optimal,horvath2008multi,qureshi2009cutting,pakbaznia2009minimizing,rao2010minimizing,stanojevic2010distributed,chase2001managing,toosi2017renewable} discussing the ``Geographical Load Balancing'' and the ``Right-sizing'' problem for data centers. The most relevant works are \cite{lin2012online,lin2013dynamic}. In \cite{lin2013dynamic}, Lin et al., considered the case of a single data center where the data center determines the workload (number of active servers) based on general convex operational cost function and linear switching cost. They proposed the ``Laze Capacity Provisioning'' online algorithm by utilizing the structure of the optimal offline solution, which achieves a competitive ratio of $3$. Later, Bansal et al. \cite{bansal_et_al:LIPIcs:2015:5297} improved the competitive ratio to $2$ by proposing a new randomized online algorithm. However, the online algorithms in both  \cite{lin2013dynamic} and \cite{bansal_et_al:LIPIcs:2015:5297} only work for the single data center case and their performance can be arbitrarily bad for the case with multiple heterogenous data centers. In \cite{lin2012online}, Lin et al., considered the ``right-sizing'' problem for the heterogeneous data center model and proposed the ``Averaging Fixed Horizon Control'' algorithm. They proved that their algorithms achieves a competitive ratio that depends on the switching cost and the convex operational cost function. In this work, we revisit the heterogeneous case and propose a new online algorithm with a better performance guarantee.

There are extensive studies on online algorithm design for cloud resource management \cite{jennings2015resource,weingartner2015cloud} and real-time dispatch \cite{zhao2016online,zheng2014online,chen2011optimal}. More specifically, \cite{buchbinder2014competitive} first introduced the concept of regularization for online algorithm design and proved that the online algorithm with a regularization term can achieve a competitive ratio proportional to $\log N$ where $N$ is number of variables. The same technique is also applied in \cite{buchbinder2014competitiveapplication} to study the problem of online restricted caching and matroid caching. In \cite{zhang2015online}, Zhang et al. investigate online resource management for Cloud-based content delivery networks. They proposed an efficient online algorithm by using the regularization technique and proved its performance guarantee. However, the cost function is linear in all these works and none of them take the switching cost offset into consideration. In this work, we consider the more general and practical case of convex operational cost and compare the performance of regularization based online algorithms and the greedy algorithms.

\section{System Model}\label{sec:model}
In this section, we discuss our system model of multiple heterogeneous data centers located in various places and the overall optimization problem.

We study a system consisting of a central load balancer and $N$ heterogeneous data centers, each with thousands of servers located in different places. The servers in each data center are assumed to be homogeneous as in previous work. The central load balancer distributes the workload to data centers and each data center either activates or deactivates a certain number of servers to serve the workload. The cost of each data center for serving the workload consists of two parts, operational cost and switching cost, both may vary across data centers.

In this paper, we consider general classes of cost functions. The operational cost is modelled by a time-dependent non-decreasing convex function $f_{i,t}(\cdot)$, where $f_{i,t}(s_i(t))$ refers to the total operational cost for data center $i$ with workload $s_i(t)$ at time $t$, which includes the energy cost for serving the workload as well as the cost associated with data transmission and delay, etc. We assume that $f_{i,t}(\cdot)$ is continuously differentiable. In \cite{lin2013dynamic}, Lin et al. provided concrete examples to show how the real data center cost can be fitted into this general convex operational cost model.

In addition to operational cost, data centers incur a switching cost when servers are switched on, which includes the energy cost of transferring server states, data migration latency, server state toggling risk, and the wear-and-tear cost~\cite{lin2013dynamic}. In addition, by reducing the number of active servers and computing resources, the user experience may be degraded, resulting a decline in revenue \cite{mao2017mobile}, which can also be captured by the switching cost. We only take into consideration the switching cost when the workload increases, incurring a cost of $\beta_i(s_i(t)-s_i(t-1))$ since turning off servers usually has a negligible cost as in \cite{wang2012characterizing,lin2013dynamic}.

We further extend the model to the situation where each data center has access to local renewable energy, or has a certain protection mechanism and delay tolerance. For instance, protection mechanisms  can reduce the wear-and-tear cost and the corresponding risk involved in the state toggling of servers~\cite{coskun2009evaluating}, while delay tolerant workload is less sensitive to the latency for toggling servers out of power-saving mode. On the other hand, timing-varying renewable energy supply can help reduce both the operational and switching costs.\footnote{In this work, we assume that the allocation of renewable energy for reducing operating cost and that for reducing switching cost follows a pre-determined scheme, and incorporate the former into the operational cost function.} To model these effects, we introduce a time-dependent offset parameter $r_i(t)$ into our model such that no switching cost is incurred when the increase of workload is less than or equal to $r_i(t)$. Such an offset parameter allows each data center to activate a certain number of servers without incurring any switching cost.

We consider a time-slotted system from $t=1$ to $t=T$, and an online setting where all the future information including the operational costs, switching costs, and the workload is unknown. The central load balancer is only aware of all the parameters at the current and past time slots. Our objective is to minimize the overall cost by dynamically dispatching the workload to each data center at the beginning of each time slot under the constraint that the total demand must be satisfied at each time slot.

At the beginning of each time-slot $t$, the workload $D(t)$, operational cost function $f_{i,t}(\cdot)$, and the offset $r_i(t)$ are revealed. The central load balancer then distributes workload $s_i(t)$ to data center $i$, incurring a operational cost $\sum_{i=1}^N f_{i,t}(s_i(t))$ and a switching cost $\sum_{i=1}^N\beta_i(s_i(t)-s_i(t-1))$ at time slot $t$. Note that we do not explicitly include a capacity constraint for each data center in (\ref{eqa:isomin}). This can be easily modeled by setting the operational cost $f_{i,t}$ to infinity when the workload assigned to data center $i$ exceeds its capacity. As long as $f_{i,t}(s)$ is continuously differentiable when $s$ is within the capacity region, all the results in this paper remain valid. The objective of the central load balancer is to minimize the overall operational cost and switching cost among all time slots as in (\ref{eqa:isomin}).
\begin{equation}\label{eqa:isomin}
\begin{aligned}
\min_{s_i(t)} \sum_{t=1}^{T} &\sum_{i=1}^N \left[f_{i,t}(s_i(t))+\beta_i(s_i(t)-s_i(t-1)-r_i(t))^+\right]\\
        s.t.\ & \sum_{i=1}^N s_i(t)\geq D(t)\ \forall t\\
              & s_i(t)\geq 0\ \forall i,t\\
              & s_i(0)=0\ \forall i
\end{aligned}
\end{equation}
where $(x)^+=\max\{0,x\}$. 

We use competitive ratio as the performance metric throughout this paper. Denote $A(1:t)$ as the input information (e.g., the workload $D(t)$, the operational cost function $f_{i,t}(\cdot)$ and the switching cost offset $r_i(t)$) from time slot $1$ to time slot $t$. For an online algorithm $\pi$, the decision $s_i(t)\ \forall i$ at each time slot $t$ can only be based on input $A(1:t)$. Let $C_{\pi}(A(1:T))$ be the total cost of algorithm $\pi$ and we compare it with the total cost of the optimal offline solution $C_{opt}(A(1:T))$ which is obtained by solving (\ref{eqa:isomin}). Then, the competitive ratio of algorithm $\pi$ is given by
\begin{displaymath}
CR_{\pi}=\max_{A(1:T)}\frac{C_{\pi}(A(1:T))}{C_{opt}(A(1:T))}
\end{displaymath}

Table~\ref{tbl:notation} summarizes the notations used in the paper.
\begin{table}[!t]
\centering
\caption{List of Notations}
\renewcommand{\arraystretch}{1.2}
\label{tbl:notation}
\small{
\begin{tabular}{|c|l|}
\hline \multicolumn{1}{|l|}{Symbol} & \multicolumn{1}{|c|}{Meaning} \\
\hline
$T$ & Number of time slots\\
\hline
$N$ & Number of data centers\\
\hline
$D(t)$ & Total workload demand at time $t$\\
\hline
$D_{max}$ & $\max_t D(t)$\\
\hline
$f_{i,t}(\cdot)$ & Operational cost of data center $i$ at time $t$\\
\hline
$\beta_i$ & Coefficient of switching cost for data center $i$\\
\hline
$\beta$ & $\max_i \beta_i$\\
\hline
$r_i(t)$ & Switching cost offset for data center $i$ at time $t$\\
\hline
$s_i(t)$ & Workload dispatched to data center $i$ at time $t$\\
\hline
\end{tabular}
}
\end{table}

\section{Online Regularization Algorithm}\label{sec:onlineconvex}
In this section, we study the right-sizing problem without switching offset, that is $r_i(t)=0\ \forall i,t$. We first review the greedy algorithm (a.k.a the AFHC algorithm in \cite{lin2012online} without look-ahead information) and its performance guarantee. Then, we present our online regularization based algorithm and compare the competitive ratios of the two algorithms.

\subsection{The Greedy Algorithm}
Lin et al. \cite{lin2012online} proposed the AFHC algorithm and analyzed the pros and cons of AFHC compared to the classic Receding Horizon Control (RHC) algorithm. They claimed that the AFHC algorithm can outperform the RHC algorithm when there are multiple heterogeneous data centers. Both AFHC and RHC work for the case with look-ahead information. When there is no look-ahead information as we consider in the paper, both algorithms reduce to the simple greedy algorithm. That is, both algorithm compute the load assignment $\tilde{s}_i(t)$ by solving the following optimization problem in each time-slot.

\begin{equation}\label{eqa:rhc}
\begin{aligned}
\tilde{s}_i(t)&=\text{argmin}_{s_i(t)} \sum_{i=1}^N\left[f_{i,t}(s_i(t))+\beta_i(s_i(t)-\tilde{s}_i(t-1))^+\right]\\
s.t.&\ \sum_{i=1}^N s_i(t) \geq D(t)\ \forall t\\
    &\ s_i(t)\geq 0\ \forall i,t
\end{aligned}
\end{equation}

Let $e_{0,i}$ denote the minimum positive constant such that $f_{i,t}(x)\geq e_{0,i} x$, $\forall x,t$. 
The following result is proved in \cite{lin2012online}:
\begin{theorem}\label{thm:rhcafhc}
The greedy algorithm is $(1+\beta/e_0)$-competitive where $\beta=\max_i\{\beta_i\}$ and $e_0=\min_i\{e_{0,i}\}$.
\end{theorem}

We note that since $f_{i,t}(x)/x$ may approach to $0$ when $x$ is close to $0$, e.g. when $f_{i,t}(x)=x^\alpha$ for $\alpha>1$, the value of $\beta/e_0$ can be huge. Moreover, our simulation results using real data from Google Cloud platform (see Section~\ref{sec:num}) indicate that the greedy algorithm may cause unnecessary frequent server switching, leading to bad performance.  
To tackle these issues, we present a new online algorithm based on a regularization technique for problem (\ref{eqa:isomin}) in the following subsection, which achieves a better competitive ratio and shows better empirical performance. 

\subsection{Online Regularization Algorithm}
Our algorithm adopts the novel framework proposed in~\cite{buchbinder2014competitive} for designing competitive online algorithms. The algorithm is essentially greedy by solving a convex optimization problem in each round, where the objective function includes both the operating cost and the regularized switching cost. As in~\cite{buchbinder2014competitive}, we use the relative entropy plus a linear term as the regularizer. The regularizer for two (discrete) distributions $\boldsymbol{\theta}$ and $\boldsymbol{u}$ is defined as $\sum_i \theta_i\ln(\theta_i/u_i)+\theta_i-u_i$.
But unlike the regularization algorithm in \cite{buchbinder2014competitive} where the operational cost function is linear and all the variables are within the range of $[0,1]$, our algorithm deals with convex functions and general non-negative domains for all the variables $s_i(t)$. We present our online algorithm in Algorithm~\ref{alg:reg2}.

\begin{algorithm}[t]
\caption{Online Regularization with Convex Operational Cost}
\label{alg:reg2}
\begin{algorithmic}[1]
\State Input: $\epsilon>0$ and $\eta=\ln(1+ND_{max}/\epsilon)$
\State Initialization: $\tilde{s}_i(0)=0$ for all $i=1,\cdots,N$
\For {$t = 1\ to\ T$}
\begin{equation}\label{eqa:regalgconv}
\hspace{-2.5ex}
\begin{aligned}
&\tilde{s}_i(t)=\text{argmin}_{s_i(t)\in P_t}\bigg\{f_{i,t}(s_i(t))\\
              &+\frac{1}{\eta}\sum_{i=1}^N\beta_i\left[(s_i(t)+\epsilon/N)\ln\left(\frac{s_i(t)+\epsilon/N}{\tilde{s}_i(t-1)+\epsilon/N}\right)-s_i(t)\right]\bigg\}
\end{aligned}
\end{equation}
where $P_t \triangleq \{s_i(t)|\sum_{i=1}^Ts_i(t)\geq D(t), s_i(t) \geq 0 \ \forall i\}$
\EndFor
\end{algorithmic}
\end{algorithm}

In Algorithm~\ref{alg:reg2}, we assume that $D_{max}=\max_t D(t)$ is known in advance and $N$ refers to the number of data centers. We compute the workload dispatch $\tilde{s}_i(t)$ by solving the convex optimization problem (\ref{eqa:regalgconv}) in each time-slot, where $\epsilon$ is a parameter that can be adjusted. Since  (\ref{eqa:regalgconv}) is a continuous convex optimization problem, it can be solved in polynomial time. Algorithm~\ref{alg:reg2} computes $\tilde{s}_i(t)$ using only the information available at the current time-slot and $\tilde{s}_i(t-1)$.

To study the performance of our online algorithm, we adopt a primal-dual analysis similar to~\cite{buchbinder2014competitive}. Below we first provide an overview of the main idea of the primal-dual technique involved in the analysis. We start with the primal problem (\ref{eqa:isomin}) with $r_i(t)=0\ \forall i,t$, which is equivalent to the following where we introduce variables $z_i(t)$ so that the objective function is continuous:

\begin{equation}\label{eqa:isomineqv}
\begin{aligned}
\min_{s_i(t)} \sum_{t=1}^{T} &\sum_{i=1}^N \left(f_{i,t}(s_i(t))+\beta_iz_i(t)\right)\\
        s.t.\ & \sum_{i=1}^N s_i(t)\geq D(t)\ \forall t\\
              & s_i(t)\geq 0\ \forall i,t\\
              & z_i(t)\geq s_i(t)-s_i(t-1)\ \forall i,t\\
              & z_i(t)\geq 0\ \forall i,t
\end{aligned}
\end{equation}

The Lagrangian function of (\ref{eqa:isomineqv}) is
\vspace{-3ex}
\begin{align*}
L(&\mu_{i,t},\lambda_t,l_{i,t},k_{i,t},s_i(t),z_i(t))\\
 &=\sum_{t=1}^{T}\sum_{i=1}^N \left[f_{i,t}(s_i(t))+\beta_iz_i(t)\right]+\sum_t \lambda_t[D(t)-\sum_i s_i(t)]\\
 &+\sum_t\sum_i\mu_{i,t}\left(s_i(t)-s_i(t-1)-z_i(t)\right)\\
 &-\sum_t\sum_i\left(l_{i,t}s_i(t)-k_{i,t}z_i(t)\right)\\
 &=\sum_t\sum_i\left[f_{i,t}(s_i(t))+(\mu_{i,t}-\mu_{i,t+1}-\lambda_t-l_{i,t})s_i(t)\right]\\
 &+\sum_t\sum_i\left(\beta_i-\mu_{i,t}-k_{i,t}\right)z_i(t)+\sum_t \lambda_tD(t)
\end{align*}
where $\lambda_t$, $l_{i,t}$, $\mu_{i,t}$ and $k_{i,t}$ are the Lagrangian multipliers for the four constraints in (\ref{eqa:isomineqv}). Thus, the dual function of (\ref{eqa:isomineqv}) is

\begin{align}
D(&\mu_{i,t},\lambda_t,l_{i,t},k_{i,t})=\min_{s_i(t),z_i(t)}L(\mu_{i,t},\lambda_t,l_{i,t},k_{i,t},s_i(t),z_i(t)) \nonumber \\
&=\min_{s_i(t)}\sum_t\sum_i\left[f_{i,t}(s_i(t))+(\mu_{i,t}-\mu_{i,t+1}-\lambda_t-l_{i,t})s_i(t)\right] \nonumber \\
&+\min_{z_i(t)}\sum_t\sum_i\left(\beta_i-\mu_{i,t}-k_{i,t}\right)z_i(t)+\sum_t \lambda_tD_t \label{eqa:dualconvex}
\end{align}

To establish a relation between the optimal offline solution and the online solution, the main idea is to assign the dual variables with values $\hat{\mu}_{i,t}$, $\hat{\lambda}_t$ $\hat{l}_{i,t}$ and $\hat{k}_{i,t}$ based on the optimal online solution $\tilde{s}_i(t)$. The weak duality tells us
\begin{displaymath}
\max_{\mu_{i,t},\lambda_t,l_{i,t},k_{i,t}}D(\mu_{i,t},\lambda_t,l_{i,t},k_{i,t}) \leq \text{value of (\ref{eqa:isomin})}
\end{displaymath}




Therefore, if we can prove
\begin{displaymath}
\text{Total Online Cost} \leq \Lambda \cdot D(\hat{\mu}_{i,t},\hat{\lambda}_t,\hat{l}_{i,t},\hat{k}_{i,t})
\end{displaymath}

\noindent for some $\Lambda>1$, then our online algorithm is $\Lambda$-compeitive, that is,  
\begin{displaymath}
\text{Total Online Cost} \leq \Lambda\cdot\text{value of (\ref{eqa:isomin})}
\end{displaymath}

All the theorems in this section and Section~\ref{sec:reg} are based on this idea. We first show that Algorithm~\ref{alg:reg2} has a smaller competitive ratio compared with the greedy algorithm.

\begin{theorem}\label{thm:oldratio}
Algorithm~\ref{alg:reg2} is $(1+\frac{\beta}{e_{0}+C})$-competitive where
\begin{equation}
C\triangleq \frac{\sum_{t=1}^T\sum_{i=1}^N \frac{\beta_i}{\eta}\ln\left(\frac{\tilde{s}_i(t)+\epsilon/N}{\tilde{s}_i(t-1)+\epsilon/N}\right)\tilde{s}_i(t)}{\sum_{t=1}^T D(t)}
\end{equation}
and $C\in [0,\beta]$.
\end{theorem}
\begin{proof}
To show $C\in [0, \beta]$, we need two inequality facts: $\sum_i a_i\log(a_i/b_i)\geq (\sum_ia_i)\log(\frac{\sum_i a_i}{\sum_i b_i})\ \ \forall a_i,b_i>0$ and $a\ln(a/b)\geq a-b\ \ \forall a,b>0$. Due to limited space, we omit the details of this part. Please refer to \cite{techreport} for a proof.

Next, we show that the competitive ratio of Algorithm~\ref{alg:reg2} is $1+\frac{\beta}{e_0+C}$. We assign dual variables $\tilde{\lambda}_t$ and $\tilde{l}_{i,t}$ to the constraints $\sum_i s_i(t)\geq D(t)\ \forall t$ and $s_i(t)\geq 0\ \forall i,t$ in (\ref{eqa:regalgconv}) respectively. Since (\ref{eqa:regalgconv}) is a convex optimization problem, by applying the KKT conditions of (\ref{eqa:regalgconv}), we have for any $i$ and $t$,
\begin{equation}\label{eqa:olregpfconvex3}
\begin{aligned}
\frac{\beta_i}{\eta}\ln\left(\frac{\tilde{s}_i(t)+\epsilon/N}{\tilde{s}_i(t-1)+\epsilon/N}\right) = \tilde{\lambda}_t-f'_{i,t}(\tilde{s}_i(t))+\tilde{l}_{i,t}
\end{aligned}
\end{equation}

\noindent By setting $\lambda_t=\tilde{\lambda}_t$, $\mu_{i,t}=\frac{\beta_i}{\eta}\ln\left(\frac{D_{max}+\epsilon/N}{\tilde{s}_i(t-1)+\epsilon/N}\right)$, $k_{i,t}=0$ and $l_{i,t}=\tilde{l}_{i,t}$, and using the fact that $\beta_i\geq \mu_{i,t}$, the dual function associated with the offline problem becomes
\begin{equation}\label{eqa:olregpfconvex7}
\begin{aligned}
&D(\mu_{i,t},\lambda_t,l_{i,t},k_{i,t})=\sum_t \lambda_tD(t)\\
&+\min_{s_i(t)}\sum_t\sum_i\left[f_{i,t}(s_i(t))+(\mu_{i,t}-\mu_{i,t+1}-\lambda_t-l_{i,t})s_i(t)\right]\\
&=\sum_t \lambda_tD(t)+\min_{s_i(t)}\sum_t\sum_i\left[f_{i,t}(s_i(t))-f'_i(\tilde{s}_i(t))s_i(t)\right]\\
&=\sum_t \lambda_tD(t)+\sum_t\sum_i\left[f_{i,t}(\tilde{s}_i(t))-f'_{i,t}(\tilde{s}_i(t))\tilde{s}_i(t)\right]
\end{aligned}
\end{equation}

\noindent where the last equation follows from the convexity of $f_{i,t}(\cdot)$. Putting (\ref{eqa:olregpfconvex3}) into (\ref{eqa:olregpfconvex7}) and using the weak duality and the fact that $\tilde{s}_i(t)\tilde{l}_{i,t}=0$, we have the following
\begin{displaymath}
\begin{aligned}
&\text{Offline Cost} \geq \text{value of (\ref{eqa:olregpfconvex7})}\\
&= \sum_{t=1}^T\sum_{i=1}^N\left[f_{i,t}(\tilde{s}_i(t))+\frac{\beta_i}{\eta}\ln\left(\frac{\tilde{s}_i(t)+\epsilon/N}{\tilde{s}_i(t-1)+\epsilon/N}\right)\tilde{s}_i(t)\right]
\end{aligned}
\end{displaymath}

Thus, the competitive ratio of Algorithm~\ref{alg:reg2} becomes
\begin{displaymath}
\begin{aligned}
&CR=\frac{\text{Online Cost}}{\text{Offline Cost}}\\
 &\leq\frac{\sum_{t=1}^T\sum_{i=1}^N \left[f_{i,t}(\tilde{s}_i(t))+\beta_i(\tilde{s}_i(t)-\tilde{s}_i(t-1))^+\right]} {\sum_{t=1}^T\sum_{i=1}^N \left[f_{i,t}(\tilde{s}_i(t))+\frac{\beta_i}{\eta}\ln\left(\frac{\tilde{s}_i(t)+\epsilon/N}{\tilde{s}_i(t-1)+\epsilon/N}\right)\tilde{s}_i(t)\right]}\\
 &\leq 1+\frac{\sum_{t=1}^T\sum_{i=1}^N \beta_i\tilde{s}_i(t)} {\sum_{t=1}^T\sum_{i=1}^N \left[f_{i,t}(\tilde{s}_i(t))+\frac{\beta_i}{\eta}\ln\left(\frac{\tilde{s}_i(t)+\epsilon/N}{\tilde{s}_i(t-1)+\epsilon/N}\right)\tilde{s}_i(t)\right]}\\
 &\leq 1+\frac{\beta \sum_{t=1}^T\sum_{i=1}^N\tilde{s}_i(t)}{\sum_{t=1}^T\sum_{i=1}^N\left[e_{0,i}\tilde{s}_i(t)+\frac{\beta_i}{\eta}\ln\left(\frac{\tilde{s}_i(t)+\epsilon/N}{\tilde{s}_i(t-1)+\epsilon/N}\right)\tilde{s}_i(t)\right]}\\
 &\leq 1+\frac{\beta}{e_0+C}
\end{aligned}
\end{displaymath}
\end{proof}

Theorem~\ref{thm:oldratio} shows that  Algorithm~\ref{alg:reg2} has a smaller competitive ratio compared to the greedy algorithm whenever $C>0$. Although it is difficult to get the accurate value of $C$ due to the complex correlation between $\epsilon$ and $\tilde{s}_i(t)$, we have $C=0$ when all $s_i(t)$ are equal and $C=\beta$ when $N=T=1$ and $\tilde{s}_i(1)=D_{max}$. Thus, the new bound in Theorem~\ref{thm:oldratio} can be close to $1+\frac{\beta}{e_0+\beta}$, which is very helpful especially when $e_0$ is very small compared to $\beta$. Further, the regularization algorithm outperforms the greedy algorithm in our real-data based simulation in Section~\ref{sec:num}.

\section{Online Regularization Algorithm with Switching Cost Offset}\label{sec:reg}
In this section, we consider the case when $r_i(t)>0$, that is, when there is a non-zero offset for the switching cost. We note that $r_i(t)$ can capture the saving from renewable energy access, the delay tolerance of computing workload, as well as the reduced server's wear-and-tear cost and state-toggling risk with certain protection mechanisms as discussed in Section~\ref{sec:model}. Further, we assume that the operational cost function is linear in this section and let $c_i(t)$ denote the unit operational cost. Problem (\ref{eqa:isomin}) then becomes
\begin{equation}\label{eqa:isomin2}
\begin{aligned}
\min_{s_i(t)} \sum_{t=1}^{T} &\sum_{i=1}^N \left[c_i(t)s_i(t)+\beta_i(s_i(t)-s_i(t-1)-r_i(t))^+\right]\\
        s.t.\ & \sum_{i=1}^N s_i(t)\geq D(t)\ \forall t\\
              & s_i(t)\geq 0\ \forall i,t\\
              & s_i(0)=0\ \forall i
\end{aligned}
\end{equation}
The problem with a general convex operational cost and a non-zero switching offset remains open.

We first note that Algorithm~\ref{alg:reg2} may perform poorly in the presence of $r_i(t)$ as shown in Section~\ref{sec:num}. Thus, we have designed a new regularization based online algorithm as shown in Algorithm~\ref{alg:reg}. Compared with Algorithm~\ref{alg:reg2}, the main difference is that Algorithm~\ref{alg:reg} distinguishes two cases when solving the convex optimization problem (\ref{eqa:regalgwithr}) based on the values of $K_c$ and $K_s$. When $1\leq K_s\leq K_c$, meaning that $r_i(t)$ is relatively small, Algorithm~\ref{alg:reg} runs the same convex optimization as in Algorithm~\ref{alg:reg2} since small $r_i(t)$ will not result in a big performance loss. When $K_s>K_c$ (e.g., when $r_i(t)$ is large), Algorithm~\ref{alg:reg} sets a different value of $\eta$ to utilize the large switching offset for a more aggressive switching policy. Note that, Algorithm~\ref{alg:reg} needs the bound of $c_i(t)$ and $r_i(t)$. As a result, the competitive ratio of Algorithm~\ref{alg:reg} also depends on these two parameters.

\begin{algorithm}[t]
\caption{Online Regularization with Linear Operational Cost and Switching Cost Offset}
\label{alg:reg}
\begin{algorithmic}[1]
\State Compute
\begin{align}
K_c&=\max\{\frac{2(1+\epsilon/D_{min})D_{max}\beta_i}{\min_{i,t}\{r_i(t)\}\min_{i,t}\{c_i(t)\}},1\} \label{K1}\\
K_s&=\frac{1}{1-\frac{\sum_i\beta_i\max_{i,t}\{r_i(t)\}}{\min_{i,t}\{c_i(t)\}D_{min}}} \label{K2}
\end{align}
\State For input $\epsilon>0$, set
\begin{displaymath}
\eta=\begin{cases}
                &\ln(1+N D_{max}/\epsilon)\ \ \text{if $1 \leq K_s\leq K_c$}\\
                &K_c\ln(1+N D_{max}/\epsilon)\ \ \text{o/w}
\end{cases}
\end{displaymath}
\State Initialize $\tilde{s}_i(0)=0$ for all $i=1,\cdots,N$
\For {$t\leftarrow\ 1\ to\ T$}
\State The ISO solves the following problem to obtain $\tilde{s}_i(t)$
\begin{equation}\label{eqa:regalgwithr}
\begin{aligned}
&\tilde{s}_i(t)=\text{argmin}_{s_i(t)\in P_t}\bigg\{c_i(t)s_i(t)\\
              &+\frac{1}{\eta}\sum_{i=1}^N\beta_i\left[(\tilde{z}_i(t)+\epsilon/N)\ln\left(\frac{\tilde{z}_i(t)+\epsilon/N}{\tilde{s}_i(t-1)+\epsilon/N}\right)-\tilde{z}_i(t)\right]\bigg\}
\end{aligned}
\end{equation}
where
$P_t=\{s_i(t)|\sum_{i=1}^Ts_i(t)\geq D(t),\ s_i(t)\geq 0\ \forall i\}$ and
\begin{displaymath}
\tilde{z}_i(t)=\begin{cases}&s_i(t)\ \ \text{if $1 \leq K_s\leq K_c$}\\
                            &\max\{s_i(t)-r_i(t),\tilde{s}_i(t-1)\}\ \ \text{o/w}
                            \end{cases}
\end{displaymath}
\EndFor
\end{algorithmic}
\end{algorithm}

\begin{theorem}\label{thm:onlinereg}
The optimal solution $\tilde{s}_i(t)$ of Algorithm~\ref{alg:reg} can achieve a competitive ratio of $\Lambda(1+2\ln(1+N\frac{D_{max}}{D_{min}}))$ compared to the offline optimal solution of (\ref{eqa:isomin2}) where
\begin{equation}\label{eqa:thmlambda}
\Lambda=\begin{cases}
K_s &\ \ \text{if $1\leq K_s \leq K_c$}\\
K_c &\ \ \text{otherwise},
\end{cases}
\end{equation}
by setting $\epsilon=D_{min}$.
\end{theorem}

The main challenge of the proof is that the offline dual function of (\ref{eqa:isomin2}) has an extra negative term that is related to $r_i(t)$, leading to the coupling of workload dispatch decisions across multiple time slots and data centers. Therefore, we prove the competitive ratio in two cases based on the value of $K_s$ and $K_c$ as in Algorithm~\ref{alg:reg}. We assign two different sets of dual variables in different cases. In case 1 where $1\leq K_s \leq K_c$, the dual variables of (\ref{eqa:regalgwithr}) are assigned to the same values as in the proof of Theorem~\ref{thm:oldratio}. 
In the other case where $K_s<1$ or $K_s>K_c$, we assign different values for the dual variables. 
Please refer to \cite{techreport} for a detailed proof.

\begin{figure}[h]
    \centering
	\includegraphics[width=0.9\linewidth]{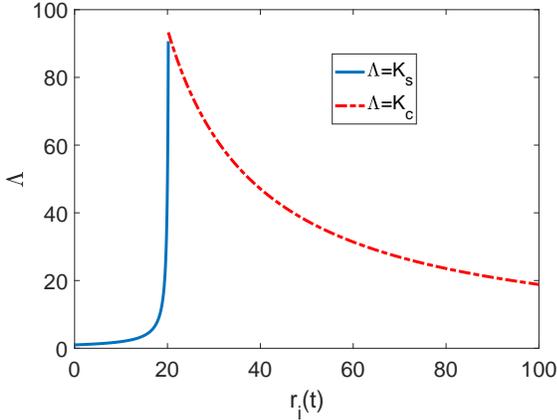}
	\caption{Competitive Ratio with Switching Offset}
	\label{fig:competitiveratio}
\end{figure}

The competitive ratio in Theorem~\ref{thm:onlinereg} depends on the values of $\beta_i$, $r_i(t)$ and $c_i(t)$, which we believe is necessary for all online algorithms. We plot $\Lambda$ versus $r_i(t)$ in Figure~\ref{fig:competitiveratio} with the same setting as in Section~\ref{sec:num}, where we have five data centers and use real data for unit electricity prices $c_i(t)$, as well as actual workload trace from the Google Cloud Platform as $D(t)$ and $\beta=6$. $r_i(t)\ \forall i,t$ in Figure~\ref{fig:competitiveratio} are all equal. The blue solid line refers to the case when $\Lambda=K_s$ and the red dash line refers to the case when $\Lambda=K_c$. Figure~\ref{fig:competitiveratio} shows that $\Lambda$ increases first and goes down after a certain value as $r_i(t)$ increases.  When $r_i(t)=0$, we have $\Lambda=1$ and the competitive ratio is $1+2\ln(1+N\frac{D_{max}}{D_{min}})$, similar to the bound in \cite{buchbinder2014competitive}. As $r_i(t)$ increases, $\Lambda$ becomes significantly large due to (\ref{K2}). When $r_i(t)$ keeps increasing, we eventually have $K_s<0$ and $K_c\approx 1$ from (\ref{K1}). Thus, $\Lambda\approx 1$ and the competitive ratio becomes close to the upper bound $1+2\ln(1+N\frac{D_{max}}{D_{min}})$ again. Thus, $\Lambda$ is determined by $K_c$ when $r_i(t)$ is large and is determined by $K_s$ when $r_i(t)$ is small as illustrated by Fig~\ref{fig:competitiveratio}. Note that if we directly apply Algorithm~\ref{alg:reg2} to (\ref{eqa:isomin2}), the performance can be very bad for large $r_i(t)$ as shown in Section~\ref{sec:num}.

\section{Numerical Results}\label{sec:num}
In this section, we evaluate the performance of our algorithms in various circumstances using real-data based simulations.

\subsection{Simulation Setup}\label{sec:numsetup}
Our simulation is based on real-world data traces for data center locations, workload, energy prices, and renewable energy supply as discussed below.

1) The workload: We use the workload trace in May 2011  from a Google Cluster of about 12.5k machines \cite{clusterdata:Wilkes2011} shown in Figure~\ref{fig:workload}. We count the average number of jobs arrived at the cluster every five minutes over two days.

2) The availability of renewable energy: We use traces with 5 minutes granularity from \cite{solardatasource,winddatasource} for solar and wind energy in five states where Google data centers are located. Figure~\ref{fig:solardata} shows the normalized Global Horizontal Irradiance (GHI) from five solar plants and Figure~\ref{fig:winddata} shows the normalized energy generation from five wind farms in the corresponding states.

The renewable energy supply determines the switching cost offset for each data center and is normalized with respect to the average workload. Let $\rho$ denote the ratio between average renewable energy supply and average workload over two days.

\begin{figure*}[ht]
  \centering
  \subfigure[Two-day average workload from Google]{
    \label{fig:workload} 
    \includegraphics[width=0.3\textwidth]{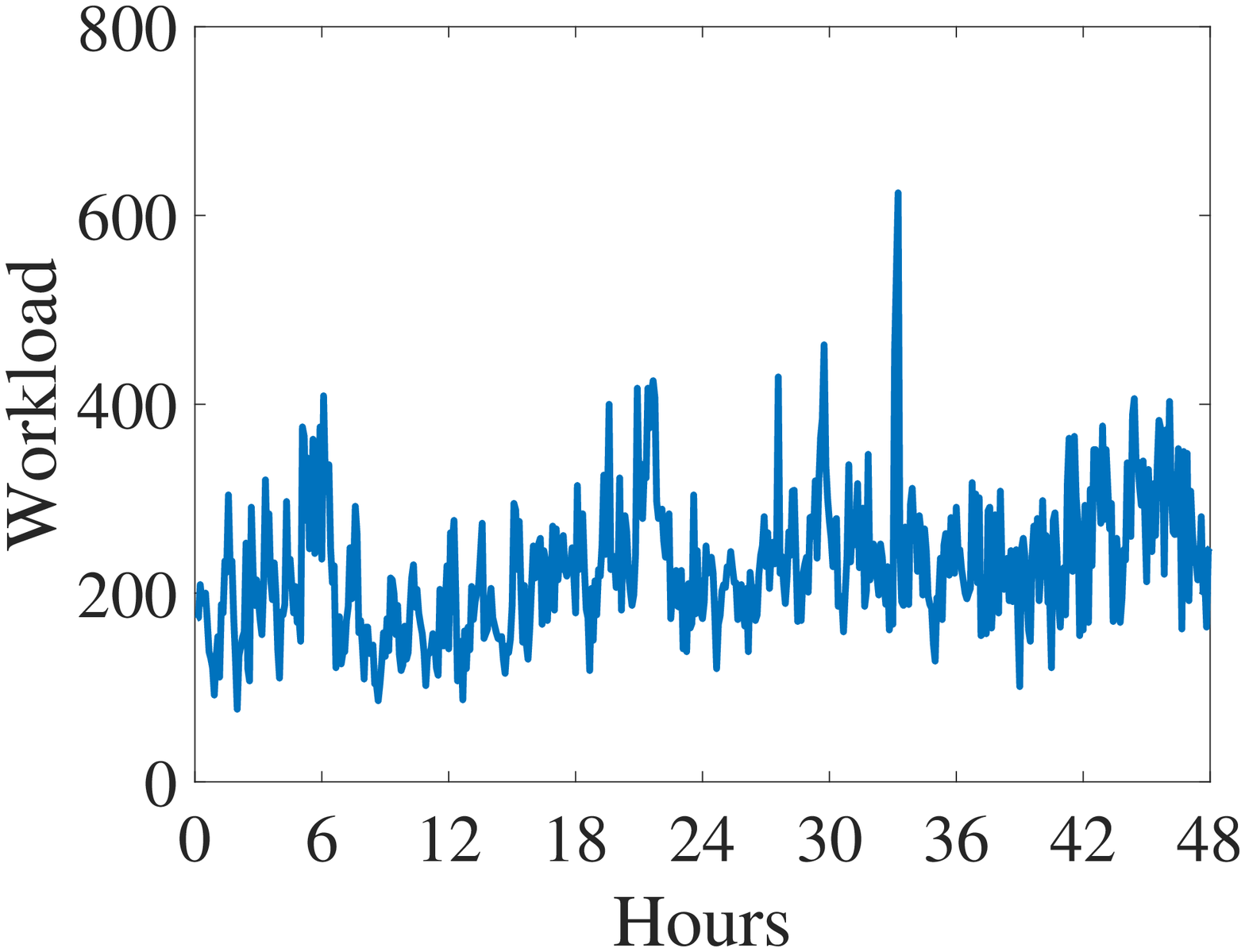}}
  \subfigure[Solar Data]{
    \label{fig:solardata} 
    \includegraphics[width=0.3\textwidth]{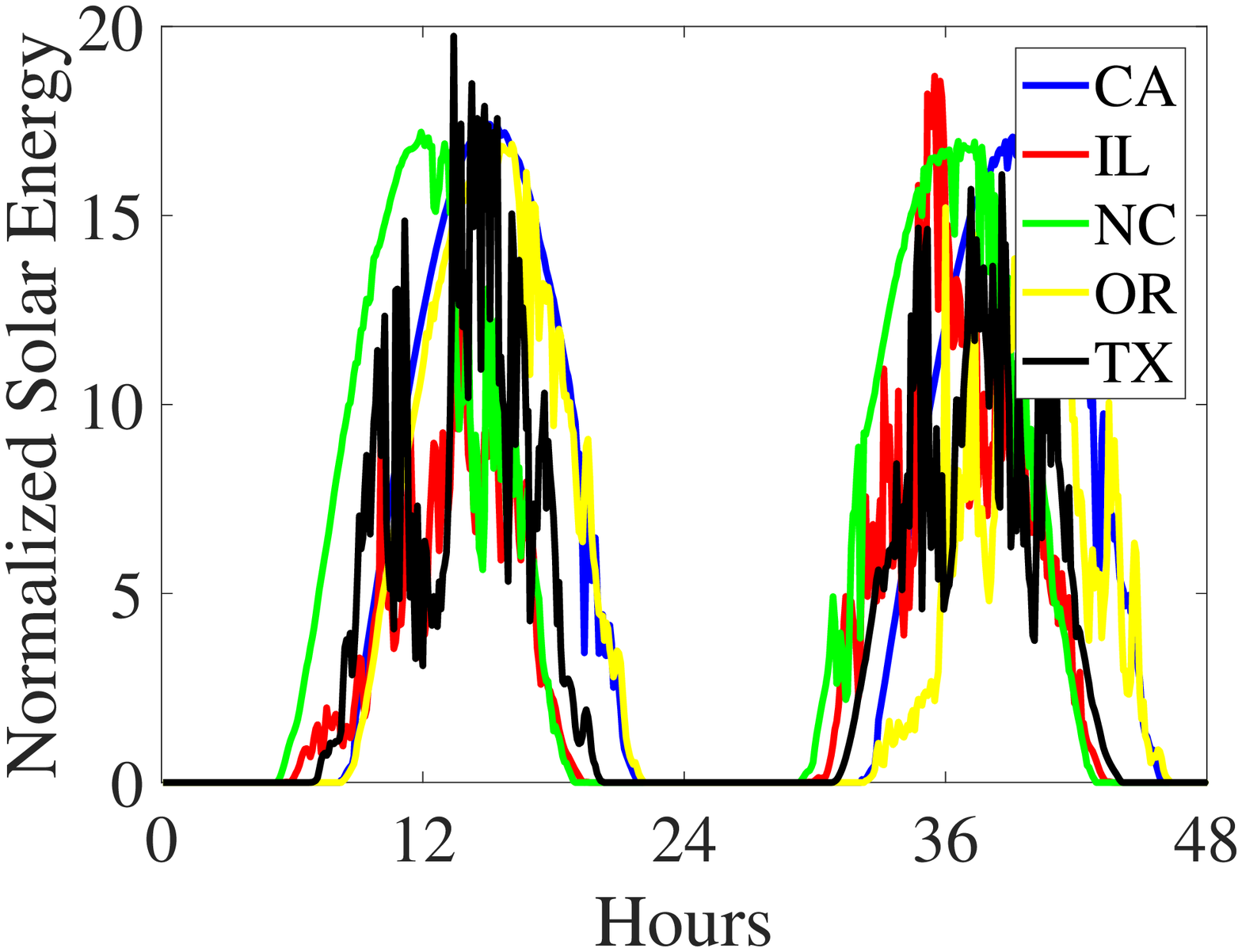}}
  \subfigure[Wind Data]{
    \label{fig:winddata} 
    \includegraphics[width=0.3\textwidth]{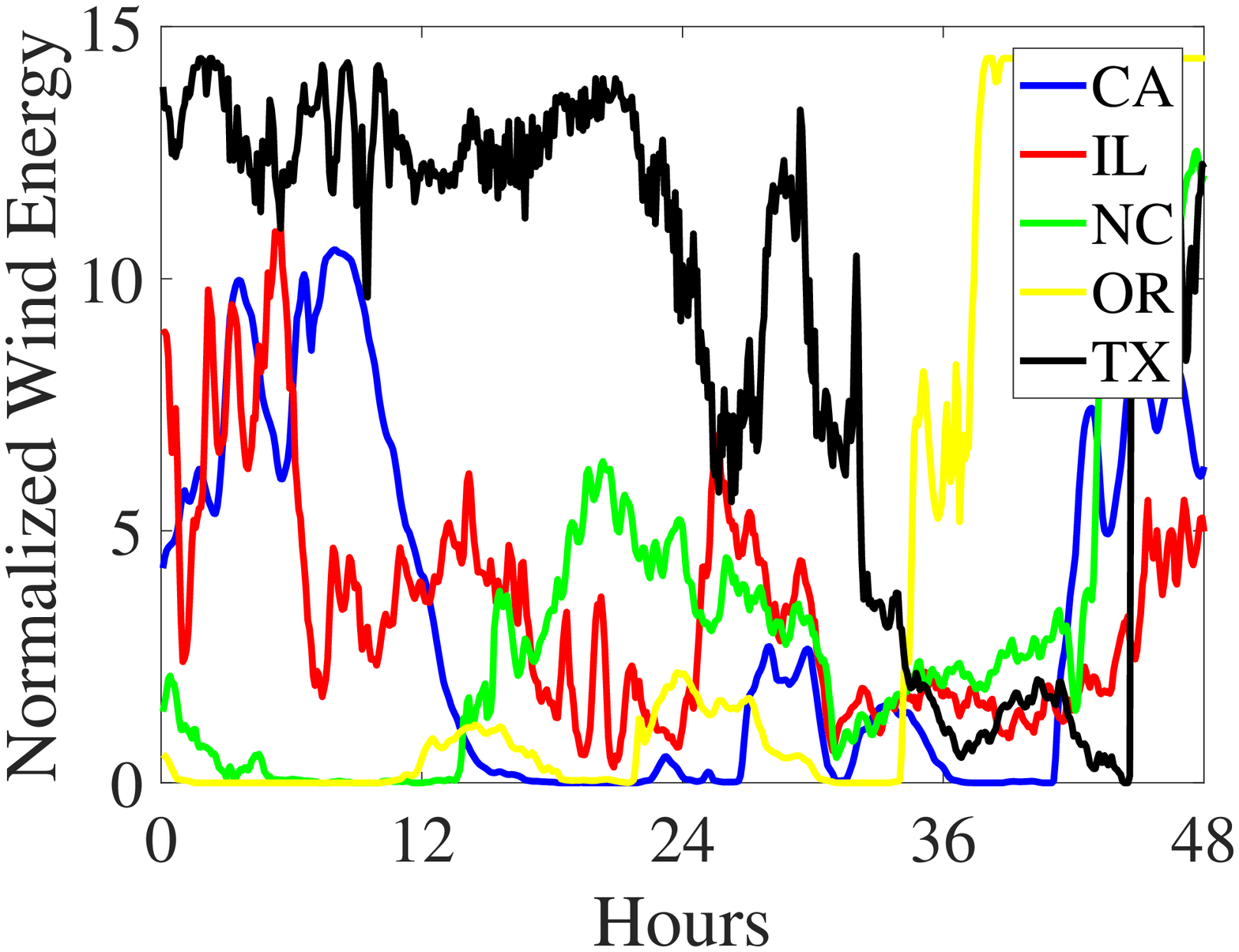}}
  \caption{Workload and the Renewable Energy Supply}
  \label{fig:renewableenergydata} 
\end{figure*}

3) The data center system: We consider a system with five data centers located in CA, IL, OR, TX and NC where Google has data centers. Each data center has access to the closet solar plant and wind turbine farm. We consider two operating cost functions in the simulation. In Figure~\ref{fig:reg_vs_greedy}, the operating cost equals to the energy price plus an extra penalty term as the following:
\begin{equation}\label{eqa:simucost}
f_{i,t}(x)=(p_i+M_{i,t})x
\end{equation}
where $p_i$ is the industrial electricity price in each of the five states in May 2011 \cite{electricityprice} and $M_{i,t}$ is a cyclic penalty term
\begin{displaymath}
M_{i,t}=\begin{cases}
              10p_i\ \ \text{if $t\mod N\geq i$}\\
              0\ \ \text{otherwise}
              \end{cases}
\end{displaymath}

In Figures~\ref{fig:change_with_beta_delay} and \ref{fig:change_with_offset_delay}, we consider another operational cost function consisting of the energy cost and delay cost. The energy cost is defined as follows:
\begin{equation}\label{eqa:simucost2}
p_i(x-\varepsilon_{i,t})^+
\end{equation}
where $p_i$ is the same price as in (\ref{eqa:simucost}) and $\varepsilon_{i,t}$ is a fixed normalized portion of the renewable energy with $\rho_o=0.2$. For the delay cost, we use a similar model as in \cite{lin2012online}:
\begin{displaymath}
D_{i,t}=\delta_i+\frac{1000ms}{\mu_i-x}
\end{displaymath}
where $\delta_i$ is the transmission delay between each data center and the central workload balancer (CA) resulting in delays between $10ms$ and $260ms$. $\mu_i=0.1(ms)^{-1}$ refers to the average number of jobs processed per unit time.

We use the renewable energy in each state as the switching offset $r_i(t)$ with normalized portion $\rho_s=0.3$ in Fig~\ref{fig:change_with_beta_delay} and vary its value in Fig~\ref{fig:change_with_offset_delay}. For switching cost, we set $\beta=20$ in Figure~\ref{fig:change_with_offset_delay} and vary $\beta$ in Figures~\ref{fig:reg_vs_greedy} and \ref{fig:change_with_beta_delay} to show its impact on total cost in different algorithms.

\subsection{Simulation Results}
We perform several simulations to evaluate the impact of the switching cost and the switching cost offset in various circumstances. In Figure~\ref{fig:reg_vs_greedy}, we set the switching cost offset $r_i(t)=0$ and compare the greedy algorithm and the regularization algorithm. In Figures~\ref{fig:change_with_beta_delay} and \ref{fig:change_with_offset_delay}, we set $r_i(t)$ to be the renewable energy supply at $t$ and investigate its effect.

\begin{figure}[h]
	\includegraphics[width=\linewidth]{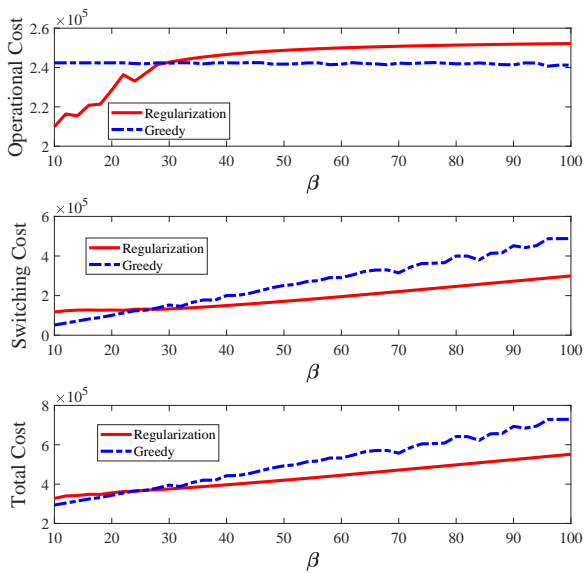}
	\caption{Regularization Algorithm vs. Greedy Algorithm}
	\label{fig:reg_vs_greedy}
\end{figure}

We first compare the performance of the greedy algorithm and the regularization based algorithm when there is no switching cost offset. In Figure~\ref{fig:reg_vs_greedy}, we vary the value of the switching cost $\beta$ while fixing all the other parameters. The top two subfigures show the operational and switching cost of both regularization algorithm and greedy algorithm respectively. The bottom one compares the overall performance of the two algorithms. Based on our analysis in Section~\ref{sec:onlineconvex}, the greedy algorithm has a larger competitive ratio. Moreover, the real performance of the regularization algorithm is also much better than the greedy algorithm when $\beta$ increases as shown in Figure~\ref{fig:reg_vs_greedy}. In addition, as $\beta$ increases, the regularization algorithm reduces the amount of workload switching (difference of workload assigned to a data center between two consecutive time slots) to each data center more dramatically compared to the greedy algorithm. The workload switching for regularization algorithm is more aggressive when $\beta$ is small and more conservative when $\beta$ is large. We can see that the switching cost of the regularization algorithm increases much slower than the greedy algorithm, leading to a higher operational cost than the latter.

\begin{figure}[h]
	\includegraphics[width=\linewidth]{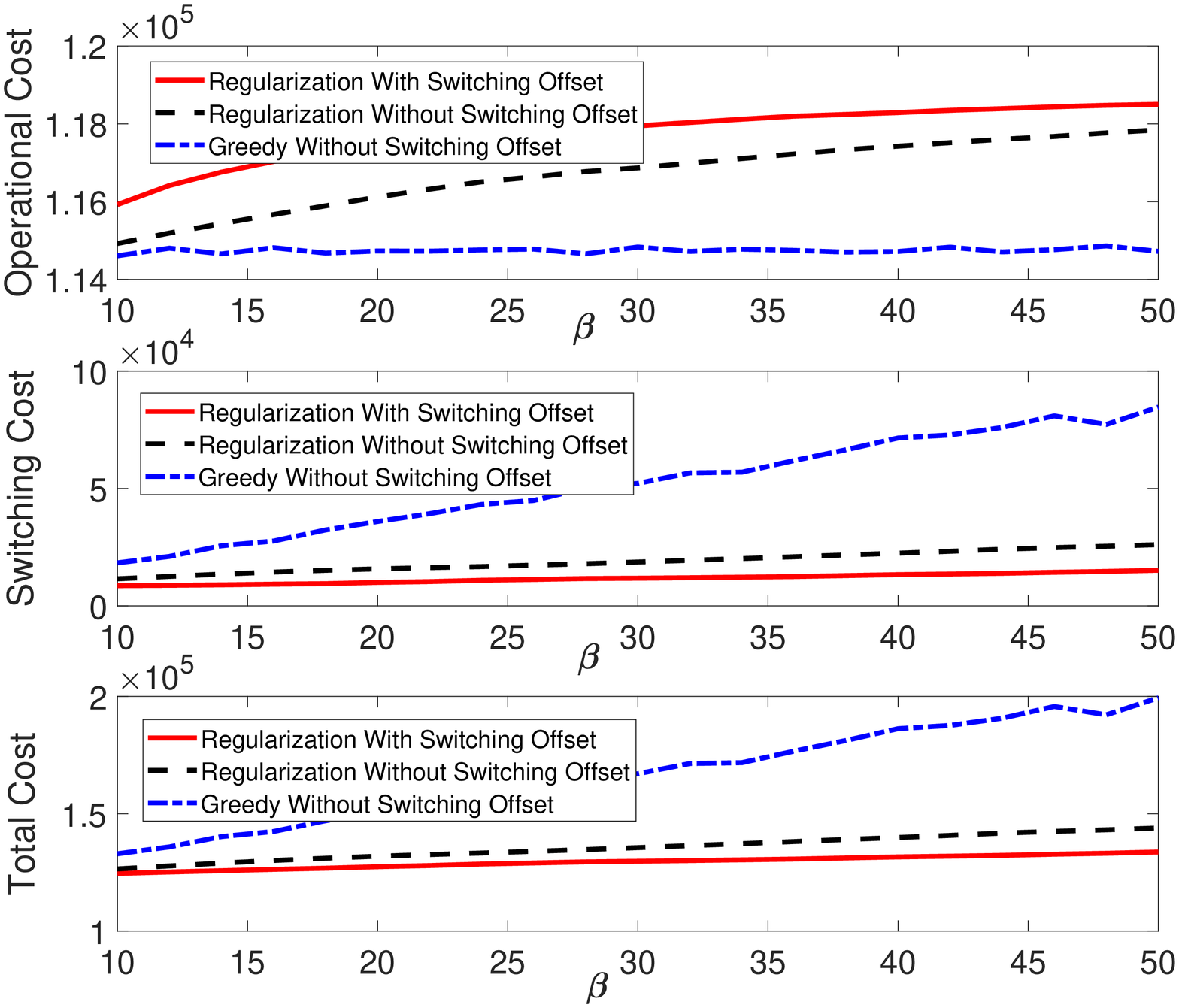}
	\caption{Cost vs. $\beta$ with Real Data}
	\label{fig:change_with_beta_delay}
\end{figure}

\begin{figure}[h]
	\includegraphics[width=\linewidth]{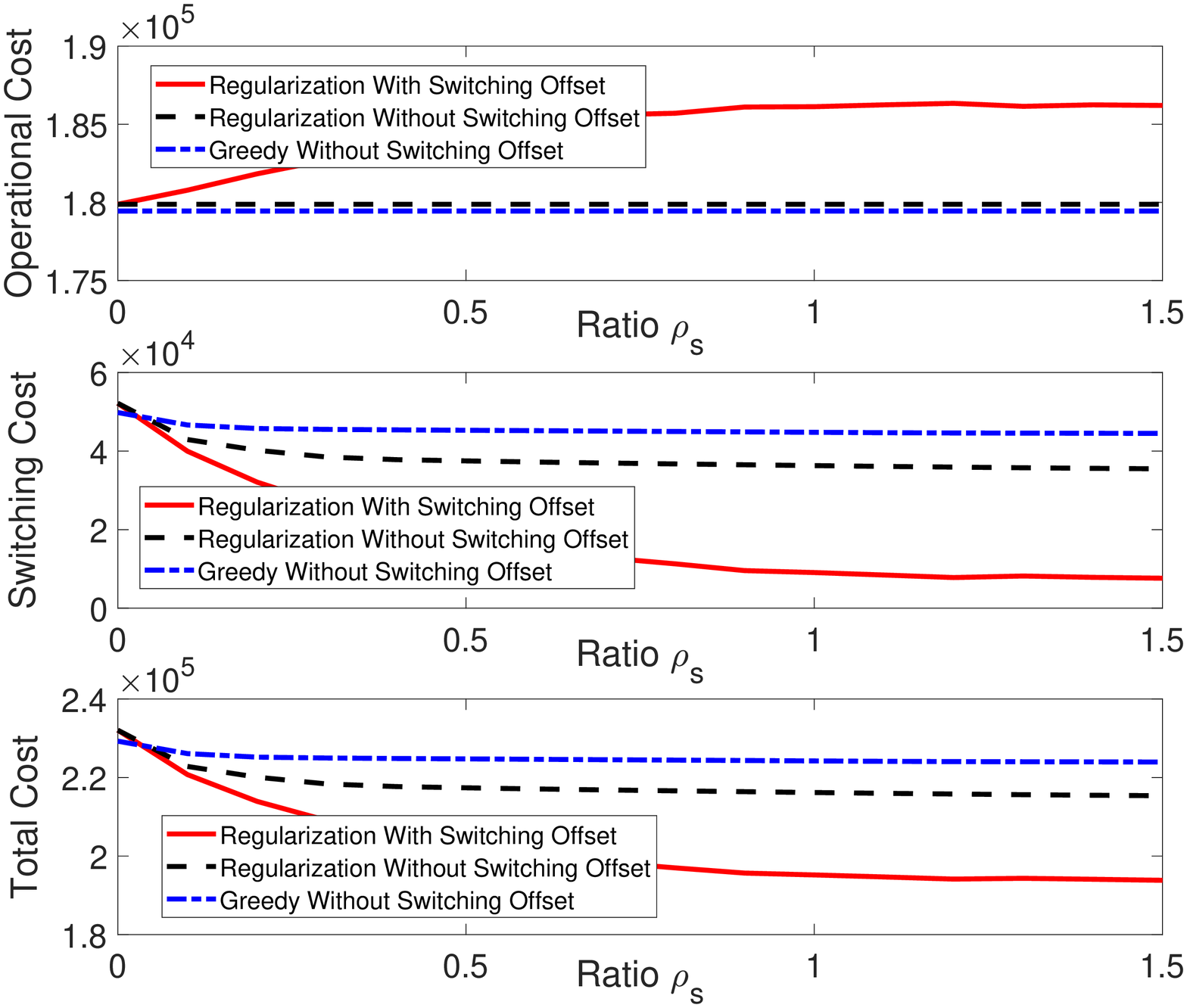}
	\caption{Cost vs. Switching Offset with Real Data}
	\label{fig:change_with_offset_delay}
\end{figure}

In Figure~\ref{fig:change_with_beta_delay} and Figure~\ref{fig:change_with_offset_delay}, we consider a fixed normalizing portion ($\rho_o=0.2$) of the renewable energy is allocated to the operational cost and the switching cost offset $r_i(t)$ is the renewable energy with normalizing parameter $\rho_s$. In both figures, we plot the operational cost, switching cost and total cost respectively. We compare the performance of Algorithm~\ref{alg:reg}, Algorithm~\ref{alg:reg2} and the greedy algorithm to investigate the impact of $\beta$ and $r_i(t)$. In Figure~\ref{fig:change_with_beta_delay}, we vary the value of $\beta$. The top two sub-figures show that the two regularization algorithms (with and without considering $r_i(t)$) have higher operational cost and smaller switching cost, which is consistent with the observation in Figure~\ref{fig:reg_vs_greedy}. The bottom sub-figure shows that Algorithm~\ref{alg:reg} outperforms the other two as expected. In Figure~\ref{fig:change_with_offset_delay}, we study the impact of $r_i(t)$ by varying the value of $\rho$, which is linearly proportional to $r_i(t)$. We observe that the total cost of Algorithm~\ref{alg:reg} is much smaller than the other two algorithms and the gap increases as $r_i(t)$ becomes larger. This is expected since Algorithm~\ref{alg:reg} utilizes $r_i(t)$ to adjust the workload dispatch more aggressively to reduce the total cost. For Algorithm~\ref{alg:reg2} and the greedy algorithm, the decrease in the total cost only comes from the increase of $r_i(t)$.

\section{Conclusion}\label{sec:conclusion}
In this paper, we study the right-sizing problem in a system consisting of a central workload balancer and multiple heterogeneous data centers with different operational cost and switching cost. We further introduce a switching cost offset to our model. Two online regularization algorithms are proposed for the case with and without the switching cost offset. For the case without switching cost offset, we show that our algorithm performs better than the greedy algorithm in terms of both the competitive ratio obtained and the real performance in real data based simulations. When considering the switching cost offset, our algorithms achieves a competitive ratio proportional to the logarithm of the number of data centers
\bibliographystyle{abbrv}
\bibliography{ref}

\begin{thebibliography}{10}

\bibitem{solardatasource}
\url{https://www.nrel.gov/rredc/}.

\bibitem{winddatasource}
\url{https://www.nrel.gov/grid/wind-integration-data.html}.

\bibitem{electricityprice}
\url{http://www.eia.doe.gov}.

\bibitem{abts2010energy}
D.~Abts, M.~R. Marty, P.~M. Wells, P.~Klausler, and H.~Liu.
\newblock Energy proportional datacenter networks.
\newblock In {\em ACM SIGARCH Computer Architecture News}, volume~38, pages
  338--347, 2010.

\bibitem{bansal_et_al:LIPIcs:2015:5297}
N.~Bansal, A.~Gupta, R.~Krishnaswamy, K.~Pruhs, K.~Schewior, and C.~Stein.
\newblock {A 2-Competitive Algorithm For Online Convex Optimization With
  Switching Costs}.
\newblock In {\em Proc. of APPROX/RANDOM}, 2015.

\bibitem{barroso2007case}
L.~A. Barroso and U.~H{\"o}lzle.
\newblock The case for energy-proportional computing.
\newblock {\em Computer}, 40(12), 2007.

\bibitem{buchbinder2014competitiveapplication}
N.~Buchbinder, S.~Chen, and J.~S. Naor.
\newblock Competitive algorithms for restricted caching and matroid caching.
\newblock In {\em European Symposium on Algorithms}, pages 209--221. Springer,
  2014.

\bibitem{buchbinder2014competitive}
N.~Buchbinder, S.~Chen, and J.~S. Naor.
\newblock Competitive analysis via regularization.
\newblock In {\em Proceedings of the Twenty-Fifth Annual ACM-SIAM Symposium on
  Discrete Algorithms}, pages 436--444, 2014.

\bibitem{chase2001managing}
J.~S. Chase, D.~C. Anderson, P.~N. Thakar, A.~M. Vahdat, and R.~P. Doyle.
\newblock Managing energy and server resources in hosting centers.
\newblock {\em ACM SIGOPS operating systems review}, 35(5):103--116, 2001.

\bibitem{chen2011optimal}
S.~Chen, L.~Tong, and T.~He.
\newblock Optimal deadline scheduling with commitment.
\newblock In {\em 49th Annual Allerton Conference on Communication, Control,
  and Computing (Allerton)}, pages 111--118, 2011.

\bibitem{chen2005managing}
Y.~Chen, A.~Das, W.~Qin, A.~Sivasubramaniam, Q.~Wang, and N.~Gautam.
\newblock Managing server energy and operational costs in hosting centers.
\newblock In {\em ACM SIGMETRICS Performance Evaluation Review}, volume~33,
  pages 303--314, 2005.

\bibitem{coskun2009evaluating}
A.~K. Coskun, R.~Strong, D.~M. Tullsen, and T.~Simunic~Rosing.
\newblock Evaluating the impact of job scheduling and power management on
  processor lifetime for chip multiprocessors.
\newblock In {\em ACM SIGMETRICS Performance Evaluation Review}, volume~37,
  pages 169--180, 2009.

\bibitem{gandhi2009optimal}
A.~Gandhi, M.~Harchol-Balter, R.~Das, and C.~Lefurgy.
\newblock Optimal power allocation in server farms.
\newblock In {\em ACM SIGMETRICS Performance Evaluation Review}, volume~37,
  pages 157--168, 2009.

\bibitem{horvath2008multi}
T.~Horvath and K.~Skadron.
\newblock Multi-mode energy management for multi-tier server clusters.
\newblock In {\em International Conference on Parallel Architectures and
  Compilation Techniques}, pages 270--279, 2008.

\bibitem{jennings2015resource}
B.~Jennings and R.~Stadler.
\newblock Resource management in clouds: Survey and research challenges.
\newblock {\em Journal of Network and Systems Management}, 23(3):567--619,
  2015.

\bibitem{koomey2011growth}
J.~Koomey.
\newblock Growth in data center electricity use 2005 to 2010.
\newblock {\em A report by Analytical Press, completed at the request of The
  New York Times}, 9, 2011.

\bibitem{lin2012online}
M.~Lin, Z.~Liu, A.~Wierman, and L.~L. Andrew.
\newblock Online algorithms for geographical load balancing.
\newblock In {\em International Green Computing Conference (IGCC)}, 2012.

\bibitem{lin2013dynamic}
M.~Lin, A.~Wierman, L.~L. Andrew, and E.~Thereska.
\newblock Dynamic right-sizing for power-proportional data centers.
\newblock {\em IEEE/ACM Transactions on Networking (TON)}, 21(5):1378--1391,
  2013.

\bibitem{mao2017mobile}
Y.~Mao, C.~You, J.~Zhang, K.~Huang, and K.~B. Letaief.
\newblock Mobile edge computing: Survey and research outlook.
\newblock {\em arXiv preprint arXiv:1701.01090}, 2017.

\bibitem{pakbaznia2009minimizing}
E.~Pakbaznia and M.~Pedram.
\newblock Minimizing data center cooling and server power costs.
\newblock In {\em Proceedings of the 2009 ACM/IEEE International Symposium on
  Low Power Electronics and Design}, pages 145--150, 2009.

\bibitem{qureshi2009cutting}
A.~Qureshi, R.~Weber, H.~Balakrishnan, J.~Guttag, and B.~Maggs.
\newblock Cutting the electric bill for internet-scale systems.
\newblock In {\em ACM SIGCOMM Computer Communication Review}, volume~39, pages
  123--134, 2009.

\bibitem{rao2010minimizing}
L.~Rao, X.~Liu, L.~Xie, and W.~Liu.
\newblock Minimizing electricity cost: Optimization of distributed internet
  data centers in a multi-electricity-market environment.
\newblock In {\em Proc. of IEEE INFOCOM}, 2010.

\bibitem{stanojevic2010distributed}
R.~Stanojevic and R.~Shorten.
\newblock Distributed dynamic speed scaling.
\newblock In {\em Proc. of IEEE INFOCOM}, 2010.

\bibitem{toosi2017renewable}
A.~N. Toosi, C.~Qu, M.~D. de~Assun{\c{c}}{\~a}o, and R.~Buyya.
\newblock Renewable-aware geographical load balancing of web applications for
  sustainable data centers.
\newblock {\em Journal of Network and Computer Applications}, 83:155--168,
  2017.

\bibitem{wang2012characterizing}
K.~Wang, M.~Lin, F.~Ciucu, A.~Wierman, and C.~Lin.
\newblock Characterizing the impact of the workload on the value of dynamic
  resizing in data centers.
\newblock In {\em ACM SIGMETRICS Performance Evaluation Review}, volume~40,
  pages 405--406, 2012.

\bibitem{weingartner2015cloud}
R.~Weing{\"a}rtner, G.~B. Br{\"a}scher, and C.~B. Westphall.
\newblock Cloud resource management: A survey on forecasting and profiling
  models.
\newblock {\em Journal of Network and Computer Applications}, 47:99--106, 2015.

\bibitem{clusterdata:Wilkes2011}
J.~Wilkes.
\newblock More {Google} cluster data.
\newblock Google research blog, Nov. 2011.
\newblock Posted at
  \url{http://googleresearch.blogspot.com/2011/11/more-google-cluster-data.html}.

\bibitem{techreport}
M.~Zhang, Z.~Zheng, and N.~Shroff.
\newblock An online algorithm for power-proportional data centers with
  switching cost.
\newblock Technical report.
\newblock
  \url{http://web.cse.ohio-state.edu/~zhang.2562/cdc-main-tech-report.pdf}.

\bibitem{zhang2015online}
X.~Zhang, C.~Wu, Z.~Li, and F.~C. Lau.
\newblock Online cost minimization for operating geo-distributed cloud cdns.
\newblock In {\em IEEE 23rd International Symposium on Quality of Service
  (IWQoS)}, pages 21--30, 2015.

\bibitem{zhao2016online}
S.~Zhao, X.~Lin, D.~Aliprantis, H.~N. Villegas, and M.~Chen.
\newblock Online multi-stage decisions for robust power-grid operations under
  high renewable uncertainty.
\newblock In {\em Proc. of IEEE INFOCOM}, 2016.

\bibitem{zheng2014online}
Z.~Zheng and N.~Shroff.
\newblock Online welfare maximization for electric vehicle charging with
  electricity cost.
\newblock In {\em Proc. of ACM e-Energy}, pages 253--263, 2014.

\end{thebibliography}
\section{Appendix}
\textbf{Proof of Theorem}~\ref{thm:oldratio}
\begin{proof}
We show that $C\in [0, \beta]$. We have the following
\begin{align*}
&\sum_{t=1}^T\sum_{i=1}^N \frac{\beta_i}{\eta}\ln\left(\frac{\tilde{s}_i(t)+\epsilon/N}{\tilde{s}_i(t-1)+\epsilon/N}\right)\tilde{s}_i(t)\\
=&\sum_{i=1}^N \frac{\beta_i}{\eta}\left[\sum_{t=1}^T\ln\left(\frac{\tilde{s}_i(t)+\epsilon/N}{\tilde{s}_i(t-1)+\epsilon/N}\right)(\tilde{s}_i(t)+\frac{\epsilon}{N})\right]\\
 &-\sum_{i=1}^N\frac{\beta_i}{\eta}\left[\frac{\epsilon}{N}\sum_{t=1}^T\ln\left(\frac{\tilde{s}_i(t)+\epsilon/N}{\tilde{s}_i(t-1)+\epsilon/N}\right)\right]\\
\mygeqa& \sum_{i=1}^N \frac{\beta_i}{\eta}\left[\left(\sum_{t=1}^T \tilde{s}_i(t)+\frac{\epsilon}{N}\right)\ln\left(\frac{\sum_{t=1}^T \tilde{s}_i(t)+\frac{\epsilon}{N}}{\sum_{t=1}^T \tilde{s}_i(t-1)+\frac{\epsilon}{N}}\right)\right]\\
 &-\sum_{i=1}^N\frac{\beta_i}{\eta}\left[\frac{\epsilon}{N}\ln\left(\frac{\tilde{s}_i(T)+\frac{\epsilon}{N}}{\tilde{s}_i(0)+\frac{\epsilon}{N}}\right)\right]\\
\mygeqb&\sum_{i=1}^N \frac{\beta_i}{\eta}\left[\sum_{t=1}^T(\tilde{s}_i(t)+\frac{\epsilon}{N})-\sum_{t=1}^T(\tilde{s}_i(t-1)+\frac{\epsilon}{N})\right]\\
 &+\sum_{i=1}^N \frac{\beta_i}{\eta}\left[(\tilde{s}_i(0)+\frac{\epsilon}{N})\ln\left(\frac{\tilde{s}_i(0)+\frac{\epsilon}{N}}{\tilde{s}_i(T)+\frac{\epsilon}{N}}\right)\right]\\
\geq & \sum_{i=1}^N\frac{\beta_i}{\eta}\left[\tilde{s}_i(T)-\tilde{s}_i(0)+\tilde{s}_i(0)-\tilde{s}_i(T)\right]=0
\end{align*}
Inequality (a) follows from the fact that
\begin{displaymath}
\sum_i a_i\log(a_i/b_i)\geq (\sum_ia_i)\log(\frac{\sum_i a_i}{\sum_i b_i})\ \ \forall a_i,b_i>0
\end{displaymath}
and inequality (b) follows from the fact that
\begin{equation}\label{eqa:lnfact}
a\ln(a/b)\geq a-b\ \ \forall a,b>0
\end{equation}
and $\tilde{s}_i(0)=0$. Thus, we have $C\geq 0$. By noting that $D_{max}\geq \tilde{s}_i(t)\geq 0$ and $\eta=\ln(1+ND_{max}/\epsilon)$, we further have
\begin{displaymath}
\begin{aligned}
C &= \frac{\sum_{t=1}^T\sum_{i=1}^N \frac{\beta_i}{\eta}\ln\left(\frac{\tilde{s}_i(t)+\epsilon/N}{\tilde{s}_i(t-1)+\epsilon/N}\right)\tilde{s}_i(t)}{\sum_{t=1}^T D(t)}\\
  &\leq \frac{\sum_{t=1}^T\sum_{i=1}^N \frac{\beta_i}{\eta}\ln\left(\frac{D_{max}+\epsilon/N}{\epsilon/N}\right)\tilde{s}_i(t)}{\sum_{t=1}^T D(t)}\\
  &= \frac{\sum_{t=1}^T\sum_{i=1}^N \beta_i\tilde{s}_i(t)}{\sum_{t=1}^T D(t)}\\
  &\leq \beta \frac{\sum_{t=1}^T\sum_{i=1}^N \tilde{s}_i(t)}{\sum_{t=1}^T D(t)}\\
  &=\beta.
\end{aligned}
\end{displaymath}
\end{proof}

\textbf{Proof of Theorem}~\ref{thm:onlinereg}
\begin{proof}
Similar to~\eqref{eqa:isomineqv}, we first convert the offline problem (\ref{eqa:isomin2}) to an equivalent problem as follows:
\begin{equation}\label{eqa:olregpf1}
\begin{aligned}
\min_{s_i(t)} & \sum_{t=1}^{T} \sum_{i=1}^N [ c_i(t)s_i(t)+\beta_ix_i(t)]\\
s.t.\ & \sum_{i=1}^N s_i(t)\geq D(t)\ \forall t=1,2,\cdots,T\\
& x_i(t)\geq s_i(t)-s_i(t-1)-r_i(t)\ \forall i,t\\
& x_i(t)\geq 0\ \forall i,t\\
& s_i(t)\geq 0\ \forall i,t\\
\end{aligned}
\end{equation}

Introducing dual variables $\lambda_t$, $\mu_{i,t}$, $k_{i,k}$ and $l_{i,t}$ to the four constraints in (\ref{eqa:olregpf1}), respectively, the Lagrangian function of (\ref{eqa:olregpf1}) is
\begin{align*}
L&=\sum_t\sum_i\left(c_i(t)+\mu_{i,t}-\mu_{i,t+1}-l_{i,t}-\lambda_t\right)s_i(t) \\
&+\sum_t\sum_i\left[(\beta_i-\mu_{i,t}-k_{i,t})x_i(t) -\mu_{i,t}r_i(t)\right] \\
&+\sum_t \lambda_tD(t)
\end{align*}

\noindent and the dual function of (\ref{eqa:olregpf1}) is
\begin{align}
D(&\mu_{i,t},\lambda_t,l_{i,t},k_{i,t})= \sum_t \lambda_tD(t)-\sum_t\sum_i\mu_{i,t}r_i(t) \nonumber \\
&+\min_{s_i(t)}\sum_t\sum_i\left(c_i(t)+\mu_{i,t}-\mu_{i,t+1}-l_{i,t}-\lambda_t\right)s_i(t) \nonumber \\
&+\min_{x_i(t)}\sum_t\sum_i(\beta_i-\mu_{i,t}-k_{i,t})x_i(t) \label{eqa:olregpf2}
\end{align}

In Algorithm~\ref{alg:reg}, we distinguish two cases based on the values of $K_s$ and $K_c$. Thus, our proof also consists of two parts. 

\noindent{\bf Case 1: $1\leq K_s\leq K_c$}. In this case, we have $\eta = \ln(1+N\frac{D_{max}}{\epsilon})$ and $\tilde{z}_i(t) = s_i(t)$ in (\ref{eqa:regalgwithr}). Let  $\tilde{\lambda}_t$ and $\tilde{l}_{i,t}$ denote the dual variables associated with the constraints $\sum s_i(t)\geq D(t)\ \forall t$ and $s_i(t)\geq 0\ \forall i,t$ in (\ref{eqa:regalgwithr}), respectively. Applying the KKT conditions to the dual problem of (\ref{eqa:regalgwithr}), we have:
\begin{equation}\label{eqa:olregpf3}
\begin{aligned}
\frac{\beta_i}{\eta}\ln\left(\frac{\tilde{s}_i(t)+\epsilon/N}{\tilde{s}_i(t-1)+\epsilon/N}\right) &= \tilde{\lambda}_t-c_i(t)+\tilde{l}_{i,t}\\
\tilde{s}_i(t)\tilde{l}_{i,t}&=0
\end{aligned}
\end{equation}

We set dual variables in the offline problem~\eqref{eqa:olregpf1} as follows: $\lambda_t=\tilde{\lambda}_t$, $\mu_{i,t}=\frac{\beta_i}{\eta}\ln\left(\frac{D_{max}+\epsilon/N}{\tilde{s}_i(t-1)+\epsilon/N}\right)$, $l_{i,t}=\tilde{l}_{i,t}$ and $k_{i,t}=0$. It follows that $\mu_{i,t}\leq \beta_i$ from the definition of $\eta$, and $\mu_{i,t}-\mu_{i,t+1} = \frac{\beta_i}{\eta}\ln\left(\frac{\tilde{s}_i(t)+\epsilon/N}{\tilde{s}_i(t-1)+\epsilon/N}\right)$. Thus, $c_i(t)+\mu_{i,t}-\mu_{i,t+1}-l_{i,t}-\lambda_t=0$ from \eqref{eqa:olregpf3}.
We then have the following lower bound for the dual function (by setting $x_i(t) = 0$ in~\eqref{eqa:olregpf2}):
\begin{align}
D(&\mu_{i,t},\lambda_t,l_{i,t},k_{i,t}) \geq \sum_t \tilde{\lambda}_tD(t)-\sum_t\sum_i\beta_ir_i(t) \label{eqa:olregpf4}
\end{align}

Denote $\mathbb{L}_t=\{i|\tilde{s}_i(t)>\tilde{s}_i(t-1)+r_i(t)\}$. We first consider the moving cost $M_t$ incurred by the online algorithm:
\begin{align*}
& M_t=\eta \sum_{i\in \mathbb{L}_t}\frac{\beta_i}{\eta}(\tilde{s}_i(t)-\tilde{s}_i(t-1)-r_i(t))\\
   &\myleqa \eta\sum_{i\in \mathbb{L}_t} (\tilde{s}_i(t)-r_i(t)+\frac{\epsilon}{N})\cdot\left(\frac{\beta_i}{\eta}\ln\left(\frac{\tilde{s}_i(t)-r_i(t)+\epsilon/N}{\tilde{s}_i(t-1)+\epsilon/N}\right)\right)\\
   &\myleqb \eta\sum_{i\in \mathbb{L}_t} (\tilde{s}_i(t)-r_i(t)+\frac{\epsilon}{N})\cdot\left(\tilde{\lambda}_t-c_i(t)\right)\\
   &\leq \eta\sum_i(\tilde{s}_i(t)+\frac{\epsilon}{N})\tilde{\lambda}_t\leq \eta(D(t)+\epsilon)\tilde{\lambda}_t\\
   &\leq \eta(1+\frac{\epsilon}{D_{min}})\tilde{\lambda}_tD(t)
\end{align*}
\noindent where (a) follows from (\ref{eqa:lnfact}) and (b) follows from (\ref{eqa:olregpf3}). We then consider the operating cost $S$ incurred by the online algorithm:
\begin{align}
S&=\sum_t\sum_ic_i(t)\tilde{s}_i(t) \nonumber \\
 &\myeqa \sum_t\sum_i\tilde{\lambda}_t\tilde{s}_i(t)-\sum_{t=1}^T\sum_{i=1}^N \frac{\beta_i}{\eta}\ln\left(\frac{\tilde{s}_i(t)+\epsilon/N}{\tilde{s}_i(t-1)+\epsilon/N}\right)\tilde{s}_i(t) \nonumber \\
 &\myleqb \sum_t\tilde{\lambda}_tD(t) \label{eqa:olregpf5}
\end{align}
\noindent where (a) follows from (\ref{eqa:olregpf3}) and (b) follows from the fact that $C\geq 0$ as proved in Theorem~\ref{thm:oldratio}. Thus, we have
\begin{displaymath}
\begin{aligned}
&\text{Total Online Cost}=S+\sum_{t=1}^T M_t\\
&\leq (1+\eta(1+\epsilon/D_{min}))\sum_t\tilde{\lambda}_t D(t)\\
&= (1+\eta(1+\epsilon/D_{min}))\frac{\sum_t\tilde{\lambda}_t D(t)-\sum_t\sum_i\beta_ir_i(t)}{\sum_t\tilde{\lambda}_t D(t)-\sum_t\sum_i\beta_ir_i(t)}\sum_t\tilde{\lambda}_t D(t)
\end{aligned}
\end{displaymath}

From~\eqref{eqa:olregpf4}, we have

\begin{align*}
\sum_t \tilde{\lambda}_tD(t)-\sum_t\sum_i\beta_ir_i(t) \leq D(\mu_{i,t},\lambda_t,l_{i,t},k_{i,t})
\end{align*}

Moreover,
\begin{align*}
&\sum_t\tilde{\lambda}_tD(t) - \sum_t\sum_i\beta_ir_i(t) \\
\mygeqa& \sum_t\sum_ic_i(t)\tilde{s}_i(t)-\sum_t\sum_i\beta_ir_i(t) \\
\mygeqb& 0
\end{align*}

\noindent where (a) follows from \eqref{eqa:olregpf5} and (b) follows from the assumption that $K_s \geq 1 > 0$. It follows that
\begin{displaymath}
\begin{aligned}
&\text{Total Online Cost}\\
&\leq (1+\eta(1+\epsilon/D_{min}))\frac{D(\mu_{i,t},l_{i,t},k_{i,t},\lambda_t)}{1-\sum_t\sum_i\beta_ir_i(t)/\sum_t\tilde{\lambda}_t D(t)}\\
&\myleqa(1+\eta(1+\epsilon/D_{min}))\frac{D(\mu_{i,t},l_{i,t},k_{i,t},\lambda_t)}{1-\sum_t\sum_i\beta_ir_i(t)/\sum_t\sum_ic_i(t)\tilde{s}_i(t)}\\
&\leq(1+\eta(1+\epsilon/D_{min}))D(\mu_{i,t},l_{i,t},k_{i,t},\lambda_t)K_s
\end{aligned}
\end{displaymath}
\noindent where (a) follows from~\eqref{eqa:olregpf5}.

\noindent{\bf Case 2: $K_s > K_c$ or $K_s<1$}. In this case, we have $\eta = K_c\ln(1+N\frac{D_{max}}{\epsilon})$. By assigning the same dual variables as in Case 1 and applying the KKT conditions to (\ref{eqa:regalgwithr}), we have
\begin{equation}\label{eqa:olregpf6}
\frac{\beta_i}{\eta}\ln\left(\frac{\tilde{z}_i(t)+\epsilon/N}{\tilde{s}_i(t-1)+\epsilon/N}\right) = \tilde{\lambda}_t-c_i(t)+\tilde{l}_{i,t}
\end{equation}


We set $\lambda_t=\tilde{\lambda}_t-\max_i\frac{\beta_i}{\eta}\ln\left(\frac{\tilde{z}_i(t)+\epsilon/N}{\tilde{s}_i(t-1)+\epsilon/N}\right)$, $\mu_{i,t}=0$, $k_{i,t}=0$ and $l_{i,t}=\tilde{l}_{i,t}$. Due to the fact that $\frac{\beta_i}{\eta}\ln\left(\frac{\tilde{z}_i(t)+\epsilon/N}{\tilde{s}_i(t-1)+\epsilon/N}\right)=0\ \forall i\notin \mathbb{L}_t$, $\tilde{l}_{i,t}=0\ \forall i\in \mathbb{L}_t$ and (\ref{eqa:olregpf3}), we must have $\lambda_t\geq 0$ and $c_i(t)+\mu_{i,t}-\mu_{i,t+1}-l_{i,t}-\lambda_t\geq 0$. Therefore, the dual function now becomes following
\begin{displaymath}
\begin{aligned}
&D(\mu_{i,t},l_{i,t},k_{i,t},\lambda_t)=\sum_t \lambda_tD(t)-\sum_t\sum_i\mu_{i,t}r_i(t)\\
&+\min_{s_i(t)}\sum_t\sum_i\left(c_i(t)+\mu_{i,t}-\mu_{i,t+1}-l_{i,t}-\lambda_t\right)s_i(t)\\
&+\min_{x_i(t)}\sum_t\sum_i(\beta_i-\mu_{i,t}-k_{i,t})x_i(t)\\
&= \sum_t \tilde{\lambda}_tD(t)-\sum_t\max_i\frac{\beta_i}{\eta}\ln\left(\frac{\tilde{z}_i(t)+\epsilon/N}{\tilde{s}_i(t-1)+\epsilon/N}\right)D(t)
\end{aligned}
\end{displaymath}
Then, for moving cost $M_t$, we have
\begin{align*}
M_t &=\eta \sum_{i\in \mathbb{L}_t}\frac{\beta_i}{\eta}(\tilde{s}_i(t)-\tilde{s}_i(t-1)-r_i(t))\\
    &\myleqa \eta\sum_{i\in \mathbb{L}_t} (\tilde{s}_i(t)-r_i(t)+\frac{\epsilon}{N})\cdot\left(\tilde{\lambda}_t-c_i(t)\right)\\
    &\leq \eta\sum_i(\tilde{s}_i(t)+\frac{\epsilon}{N})\tilde{\lambda}_t-\sum_{i\in \mathbb{L}_t}\eta \tilde{\lambda}_t r_i(t)\\
    &\myleqb \eta(1+\frac{\epsilon}{D_{min}})\tilde{\lambda}_tD_t\\
    &-\sum_{i\in \mathbb{L}_t}\eta r_i(t)\left[\frac{\beta_i}{\eta}\ln\left(\frac{\tilde{s}_i(t)-r_i(t)+\epsilon/N}{\tilde{s}_i(t-1)+\epsilon/N}\right)+c_i(t)\right]
\end{align*}
where (a) follows from the same argument in Case 1 and (b) follows from (\ref{eqa:olregpf6}) and $\tilde{l}_{i,t}=0\ \forall i\in \mathbb{L}_t$. For the operating cost, we have
\begin{displaymath}
\begin{aligned}
S&=\sum_t\sum_ic_i(t)\tilde{s}_i(t)\\
 &\myeqa \sum_t\sum_i\tilde{\lambda}_t\tilde{s}_i(t)-\sum_{t=1}^T\sum_{i=1}^N \frac{\beta_i}{\eta}\ln\left(\frac{\tilde{z}_i(t)+\epsilon/N}{\tilde{s}_i(t-1)+\epsilon/N}\right)\tilde{s}_i(t)\\
 &\myleqb \sum_t\tilde{\lambda}_t D(t)
\end{aligned}
\end{displaymath}
where (a) follows (\ref{eqa:olregpf6}) and (b) follows the fact that $\tilde{z}_i(t)\geq \tilde{s}_i(t-1)$.  By denoting
\begin{displaymath}
P=\eta\sum_t\sum_{i\in \mathbb{L}_t}r_i(t)\left[c_i(t)+\frac{\beta_i}{\eta}\ln\left(\frac{\tilde{s}_i(t)-r_i(t)+\epsilon/N}{\tilde{s}_i(t-1)+\epsilon/N}\right)\right]
\end{displaymath}
we have
\begin{align*}
&\text{Total Online Cost}=S+\sum_{t=1}^T M_t\\
&\leq (1+\eta(1+\epsilon/D_{min}))\sum_t\tilde{\lambda}_tD(t)-P\\
&\leq (1+\eta(1+\epsilon/D_{min}))\left(\sum_t\tilde{\lambda}_tD(t)-\frac{P}{1+\eta(1+\epsilon/D_{min})}\right)\\
&\myleqa (1+\eta(1+\epsilon/D_{min}))\\
&\ \ \cdot \left[\sum_t\tilde{\lambda}_tD(t)-\sum_t\max_i\frac{\beta_i}{\eta}\ln\left(\frac{\tilde{z}_i(t)+\epsilon/N}{\tilde{s}_i(t-1)+\epsilon/N}\right)D(t)\right]\\
&\leq K_c\left(1+(1+\epsilon/D_{min})\ln\left(1+ND_{max}/\epsilon\right)\right)D(\mu_{i,t},l_{i,t},k_{i,t},\lambda_t)
\end{align*}
where (a) follows from
\begin{align*}
K_c &= \max\{\frac{2(1+\epsilon/D_{min})D_{max}\beta_i}{\min_{i,t}\{r_i(t)\}\min_{i,t}\{c_i(t)\}},1\}\\
&\geq \frac{2(1+\epsilon/D_{min})D_{max}\beta_i}{r_i(t)c_i(t)} \ \forall i,t \\
\Rightarrow &\frac{r_i(t)}{2(1+\epsilon/D_{min})}\left[\frac{c_i(t)K_c}{\beta_i}+1\right]\geq D(t)\ \forall i,t\\
\Rightarrow &\frac{r_i(t)}{2(1+\epsilon/D_{min})}\left[\frac{c_i(t)K_c\ln(1+N\frac{D_{max}}{\epsilon})}{\beta_i\ln(1+N\frac{D_{max}}{\epsilon})}+1\right]\geq D(t)\ \forall i,t\\
\Rightarrow &\frac{r_i(t)}{2(1+\epsilon/D_{min})}\left[\frac{c_i(t)\eta}{\beta_i\ln(1+N\frac{D_{max}}{\epsilon})}+1\right]\geq D(t)\ \forall i,t\\
\Rightarrow &\frac{\eta r_i(t)}{2\eta(1+\epsilon/D_{min})}\bigg[\frac{c_i(t)\eta}{\beta_i}\frac{\beta_i}{\eta}\ln\left(\frac{\tilde{z}_i(t)+\epsilon/N}{\tilde{s}_i(t-1)+\epsilon/N}\right)/\ln(1+N\frac{D_{max}}{\epsilon})\\
&+\frac{\beta_i}{\eta}\ln\left(\frac{\tilde{s}_i(t)-r_i(t)+\epsilon/N}{\tilde{s}_i(t-1)+\epsilon/N}\right)\bigg]\\
&\geq \frac{\beta_i}{\eta}\ln\left(\frac{\tilde{z}_i(t)+\epsilon/N}{\tilde{s}_i(t-1)+\epsilon/N}\right)D(t) \ \ \forall i\in \mathbb{L}_t,t \\
\Rightarrow &\frac{\eta r_i(t)}{1+\eta(1+\epsilon/D_{min})}\left[c_i(t)+\frac{\beta_i}{\eta}\ln\left(\frac{\tilde{s}_i(t)-r_i(t)+\epsilon/N}{\tilde{s}_i(t-1)+\epsilon/N}\right)\right]\\
&\geq \frac{\beta_i}{\eta}\ln\left(\frac{\tilde{z}_i(t)+\epsilon/N}{\tilde{s}_i(t-1)+\epsilon/N}\right)D(t) \ \ \forall i\in \mathbb{L}_t,t \\
\Rightarrow & \frac{\eta\sum_t\sum_{i\in \mathbb{L}_t}r_i(t)\left[c_i(t)+\frac{\beta_i}{\eta}\ln\left(\frac{\tilde{s}_i(t)-r_i(t)+\epsilon/N}{\tilde{s}_i(t-1)+\epsilon/N}\right)\right]}{1+\eta(1+\epsilon/D_{min})}\\
&\geq \sum_t\sum_{i\in \mathbb{L}_t}\frac{\beta_i}{\eta}\ln\left(\frac{\tilde{z}_i(t)+\epsilon/N}{\tilde{s}_i(t-1)+\epsilon/N}\right)D(t)\\
\myrightarrowa &\frac{P}{1+\eta(1+\epsilon/D_{min})}\geq \sum_t\max_i\frac{\beta_i}{\eta}\ln\left(\frac{\tilde{z}_i(t)+\epsilon/N}{\tilde{s}_i(t-1)+\epsilon/N}\right)D(t)
\end{align*}
where (a) holds by following reasons:
\begin{enumerate} 
\item $j_t\in \mathbb{L}_t$ where $j_t=\text{argmax}_i\frac{\beta_i}{\eta}\ln\left(\frac{\tilde{z}_i(t)+\epsilon/N}{\tilde{s}_i(t-1)+\epsilon/N}\right)$ $\forall t$
\item $\frac{\beta_i}{\eta}\ln\left(\frac{\tilde{z}_i(t)+\epsilon/N}{\tilde{s}_i(t-1)+\epsilon/N}\right)=0$ $\forall i\notin \mathbb{L}_t$ $\forall t$
\end{enumerate}
Then, combining the two cases above, we have (\ref{eqa:thmlambda}).
\end{proof}

\end{document}